\documentclass[10pt]{article}
\usepackage{hyperref}
\usepackage{amsmath,amssymb,amsthm}
\usepackage{geometry}
\usepackage{graphicx}
\usepackage{booktabs}
\usepackage{chngcntr}
\usepackage{url}
\usepackage[compatibility=false, small, hang, nooneline]{caption} 
\usepackage{subcaption}
\usepackage[boxruled, lined, linesnumbered]{algorithm2e}
\usepackage[colorinlistoftodos, textwidth=3.75cm,]{todonotes}
\usepackage{rotating}
\expandafter\ifx\csname DontPrintSemicolon\endcsname\relax
\newcommand{\DontPrintSemicolon}{\dontprintsemicolon}\fi
\newcommand{\plll}{PotLLL\xspace}
\newcommand{\pllltwo}{PotLLL2\xspace}
\newcommand{\dlll}{DeepLLL\xspace}
\newcommand{\ntl}{NTL\xspace}

\newcommand{\tL}{\mathcal{L}}

\newcommand{\pot}{\mathrm{Pot}}

\newcommand{\argmin}{\mathrm{argmin}}

\newcommand{\R}{\mathbb{R}}
\newcommand{\Q}{\mathbb{Q}}
\newcommand{\Z}{\mathbb{Z}}

\newcommand{\spa}{\mathrm{span}\,}

\newcommand{\vol}{\mathrm{vol}\,}

\newcommand{\norm}[1]{\lVert #1 \rVert}

\newcommand{\grossO}[1]{\ensuremath{\mathcal{O}(#1)}}

\newtheorem{definition}{Definition}[section]

\newtheorem{proposition}[definition]{Proposition}
\newtheorem{lemma}[definition]{Lemma}
\newtheorem{corollary}[definition]{Corollary}

\begin{document}

\title{\plll: A Polynomial Time Version of LLL With Deep Insertions}

\date{\today}

\author{Felix Fontein\footnote{Universit\"at Z\"urich, \href{mailto:felix.fontein@math.uzh.ch}{felix.fontein@math.uzh.ch}} 
   \and Michael Schneider\footnote{Technische Universit\"at Darmstadt, \href{mailto:mischnei@cdc.informatik.tu-darmstadt.de}{mischnei@cdc.informatik.tu-darmstadt.de}} 
   \and Urs Wagner\footnote{Universit\"at Z\"urich, \href{mailto:urs.wagner@math.uzh.ch}{urs.wagner@math.uzh.ch}}}

\maketitle

\begin{abstract}
Lattice reduction algorithms have numerous applications in number theory, algebra, as well as in cryptanalysis.
The most famous algorithm for lattice reduction is the LLL algorithm. In polynomial time it computes a reduced basis with provable output quality.
One early improvement of the LLL algorithm was LLL with deep insertions (\dlll).
The output of this version of LLL has higher quality in practice but the running
time seems to explode. Weaker variants of \dlll, where the insertions are
restricted to blocks, behave nicely in practice concerning the running time.
However no proof of polynomial running time is known.
In this paper \plll, a new variant of \dlll with provably polynomial running time, is
presented. We compare the practical behavior of the new algorithm to classical
LLL, BKZ as well as blockwise variants of \dlll regarding both the output
quality and running time.

\vspace{3mm}

\end{abstract}

\section{Introduction}\label{sec:introduction}

The well-known LLL lattice reduction algorithm was presented in 1982 by Lenstra, Lenstra, Lov\'asz
\cite{LLL}.  Apart from various other applications (e.g.~\cite[Chapter 9,10]{ng10}) it has already
at an early stage been used to attack various public key cryptosystems.  Nevertheless lattice
problems remain popular when it comes to the construction of provably secure cryptosystems
(e.g.~\cite[Chapter~13]{ng10}).  Consequently improvements in lattice reduction still have a direct
impact on the security of many cryptosystems and rise high interest in the crypto-community.

Many lattice reduction algorithms used in practice are generalizations of the LLL algorithm.  The
Block-Korkine-Zolotarev (BKZ) reduction algorithm by Schnorr and Euchner~\cite{BKZ} is probably the
most used algorithm when stronger reduction than the one achieved by LLL is required. It can be seen
as a generalization of LLL to higher blocksizes, and while the running time seems to behave well for
small blocksizes~\cite{gam08}, no useful upper bound has been proven so far.  Another improvement of
the LLL algorithm has also been suggested in~\cite{BKZ}.  While in LLL adjacent basis vectors are
swapped if certain conditions are satisfied, in the so called LLL with deep insertions (\dlll in the
sequel), basis vectors can be swapped even when not adjacent. The practical behavior of~\dlll when
it comes to the reducedness of the output basis is superior the one of LLL. Unfortunately also the
running time explodes and does not seem to be polynomial in the dimension of the lattice. One
attempt to get across this problem is to restrict the insertions to certain blocks of basis
vectors. While the authors in \cite{BKZ} claim that these blockwise restriction variants of~\dlll
run in polynomial time, we are not aware of any proof thereof.  For an overview on the practical
behavior of the different variants and improvements on LLL, we refer to~\cite{ng06,gam08}.  There
the practical behavior of the reduction algorithms is investigated using the widely used~\ntl
library.

In this paper we present two new versions of~\dlll, called~\plll and \pllltwo. To our knowledge it
is the first improvement of LLL with regard to deep insertions which provably runs in polynomial
time. The practical behavior of~\plll and \pllltwo regarding both the output quality and running
time is empirically tested and compared to BKZ and~\dlll with different blocksizes.  The tests are
performed with a completely new implementation of the different reduction algorithms. This
additionally allows an independent review of the results in~\cite{ng06,gam08}. The tests indicate
that our algorithm can serve as a serious alternative to BKZ with low blocksizes.  This paper is an
extension of the work presented at WCC~2013 in Bergen~\cite{fo12}.

The paper is organized as follows. In Section~\ref{sec:prelim} all necessary notations and
definitions are given.  In Section~\ref{sec:algorithm} the reduction notion and the new algorithm is
presented and a theoretical analysis is provided.  Section~\ref{sec:experiments} contains the
empirical results and conclusions are drawn in Section~\ref{sec:conclusion}.

\section{Preliminaries}
\label{sec:prelim}

 A lattice $\tL \subset \R^m$ of rank $n$ and dimension $m$ is a discrete subgroup of $\R^m$
generated by integer
    linear combinations of $n$ linearly independent vectors $b_1,\dots,b_n$ in
$\R^m$:
    \[
    \tL= \tL(b_1,\dots,b_n) := \biggl\{\sum_{i=1}^n x_ib_i \biggm| \forall i :
x_i \in \Z \biggr\}\,.
    \]
We will often write the basis $b_1,\dots,b_n$ as rows of a matrix $B$ in the following way $B= [b_1,\dots,b_n]$.
In order to have exact representations in computers, only lattices in $\Q^n$ are
considered. Simple scaling by the least common multiple of the denominators
allows us to restrict ourselves to integer lattices $\tL \subseteq \Z^m$.
The volume of a lattice $\tL(B)$ equals the volume of its fundamental
parallelepiped $\vol(\tL)=\sqrt{\det(BB^t)}$.
For $n\geq 2$, a lattice has infinitely many bases as $\tL(B)=\tL(B')$ if and only if $\exists U \in
GL_n(\Z): B=UB'$. Therefore, the volume of a lattice is well defined.
By $\pi_k: \R^m \rightarrow \spa\{b_1, \dots, b_{k-1} \}^\bot$ we denote the
orthogonal projection from $\R^m$
  onto the orthogonal complement of $\spa\{b_1, \dots, b_{k-1} \}$.
 In
particular, $\pi_1 =
  \mathrm{id}_{\R^m}$ and $b^*_i:=\pi_i(b_i)$ equals the $i$-th basis vector of the Gram-Schmidt orthogonalization 
$B^*=[b^*_1,\dots, b^*_n]$ of $B$.
By $\mu_{i,j}:= \langle b_i,b^*_j \rangle / \langle b^*_j,b^*_j \rangle$,
$j<i$, we denote the Gram-Schmidt coefficients. The Gram-Schmidt vectors can
iteratively be computed by
 $\pi_i(b_i)=b_i^* = b_i - \sum_{j=1}^{i-1} \mu_{i,j} b_j^*$.

Throughout this paper, by $\norm{\cdot}$ we denote the Euclidean norm and by
$\lambda_1(\tL)$ we denote the length of a shortest non-zero vector in $\tL$
with respect to the Euclidean norm: $\lambda_1(\tL):=\min_{v\in
\tL\setminus\{0\}} \norm{v}$.
Determining $\lambda_1(\tL)$ is commonly known as the shortest vector problem
(SVP) and is proven to be NP-hard (under randomized reductions) (see e.g.
\cite{mi02}). 
Upper bounds with respect to the determinant exist, for all rank $n$ lattices $\tL$ we have  \cite{ng10}
\[
\frac{\lambda_1(\tL)^2}{\vol(\tL)^{2/n}}\leq \gamma_n \leq 1+\frac{n}{4}\,,
\]
where $\gamma_n$ is the \emph{Hermite constant} in dimension $n$.
Given a relatively short vector $v \in \tL$, one measures its quality by the
\emph{Hermite factor} $\norm{v}/\vol(\tL)^{1/n}$ it achieves.
Modern lattice reduction algorithms achieve a Hermite factor which is
exponential in $n$ and no polynomial time algorithm is known to achieve linear
or polynomial Hermite factors.

 Let $S_n$ denote the group of permutations of $n$ elements. By applying $\sigma \in S_n$ to a basis
 $B=[b_1,\dots,b_n]$, the basis vectors are reordered $\sigma
 B=[b_{\sigma(1)},\dots,b_{\sigma(n)}]$. For $1\leq k \leq \ell \leq n$ we define a class of
 elements $\sigma_{k,\ell} \in S_n$ as follows:
\begin{equation}
\sigma_{k,\ell}(i)=\left\{
\begin{array}{lll}
i & \mbox{for} &  i<k \mbox{ or } i>\ell\,, \\
\ell & \mbox{for} & i=k\,, \\
i-1 & \mbox{for} & k<i\leq \ell\,.
\end{array}
\right.
\end{equation}
Note that $\sigma_{k,\ell} = \sigma_{k,k+1} \sigma_{k+1,k+2} \cdots \sigma_{\ell-1,\ell}$ and that
$\sigma_{k,k+1}$ is swapping the two elements~$k$ and $k + 1$.

  \begin{definition}
    Let $\delta \in \left(1/4, 1\right]$.  A basis $B=[b_1,\dots,b_n]$
of a lattice
    $\tL(b_1,\dots,b_n)$ is called \emph{$\delta$-LLL reduced} if and only if it
satisfies the
    following two conditions:
    \begin{enumerate}
      \item $\forall 1\leq j<i\leq n: |\mu_{i,j}| \leq
\frac{1}{2}$ (size-reduced).
      \item  $1\leq k <n : \delta \cdot \|\pi_k(b_k)\|^2 \leq
\|\pi_k(b_{k+1})\|^2$ (Lov{\'a}sz-condition).
    \end{enumerate}
  \end{definition}
  A $\delta$-LLL reduced basis $B=[b_1,\dots,b_n]$ can be computed in polynomial time \cite{LLL} and  provably satisfies the following
bounds:
\begin{equation}
\|b_1\| \leq \bigl(\delta-1/4\bigr)^{-(n-1)/2} \cdot \lambda_1(\tL(B)) \quad \text{and} \quad
\|b_1\| \leq \bigl(\delta-1/4\bigr)^{-(n-1)/4} \cdot \vol(\tL(B))^{1/n} \label{equ:hermiteLLL}.
\end{equation}
  While these bounds can be reached, they are worst case bounds. In practice,
 LLL reduction algorithms behave much better \cite{ng06}.
One early attempt to improve the LLL reduction algorithm is due to Schnorr and
Euchner \cite{BKZ} who came up with the notion of a \dlll reduced basis:
\begin{definition}
    Let $\delta \in \left(1/4, 1\right]$.  A basis $B=[b_1,\dots,b_n]$
of a lattice
    $\tL(b_1,\dots,b_n)$ is called \emph{$\delta$-\dlll reduced with blocksize $\beta$} if and only if
it satisfies the
    following two conditions:
    \begin{enumerate}
      \item $\forall 1\leq j<i\leq n: |\mu_{i,j}| \leq \frac{1}{2}$ 
(size-reduced).
      \item  $\forall 1\leq k < \ell \leq n \mbox{ with } k\leq \beta \vee
\ell-k \leq \beta: \delta \cdot \|\pi_k(b_k)\|^2 \leq\|\pi_k(b_{\ell})\|^2$.
    \end{enumerate}
\end{definition}
If $\beta=n$ we simply call this a \dlll reduced basis.
While the first basis vector of \dlll reduced bases in the worst case does not
achieve a better Hermite factor than classical LLL (see Section
\ref{sec:critical}), the according reduction algorithms usually return much
shorter vectors than pure LLL.
Unfortunately no polynomial time algorithm to compute \dlll reduced bases is known. 

The following definition is used in the proof (see e.g. \cite{mi02}) of the polynomial running time of the LLL reduction algorithm and will play a main role in our improved variant of LLL.
\begin{definition}
The \emph{potential} $\pot(B)$ of a lattice basis $B=[b_1,\dots,b_n]$ is defined as
\[
\pot(B):=\prod^n_{i=1} \vol(\tL(b_1,\dots,b_i))^2=\prod^n_{i=1}
\|b^*_i\|^{2(n-i+1)}\,.
\]
\end{definition}
Here it is used that $\vol(\tL) = \prod_{i=1}^{n} \norm{b_i^*}$.
Note that, unlike the volume of the lattice, the potential of a basis is variant
under basis permutations. The following lemma describes how the potential
changes if $\sigma_{k,\ell}$ is applied to the basis.

\begin{lemma}
\label{lem:pot}
Let $B=[b_1,\dots,b_n]$ be a lattice basis. Then for $1\leq k \leq \ell \leq n$
\[
\pot(\sigma_{k,\ell}B)=\pot(B) \cdot \prod^{\ell}_{i=k} \frac{\|\pi_i(b_{\ell})\|^2}{\|\pi_i(b_i)\|^2}.
\]
\end{lemma}
\begin{proof}
  First note that it is well-known that $\pot(\sigma_{k,k+1}B)= \|\pi_k(b_{k+1})\|^2 /
  \|\pi_k(b_k)\|^2 \cdot \pot(B)$. This property is used in the proofs of the polynomial running
  time of LLL \cite{LLL,mi02}.

  We prove the claim by induction over $\ell-k$.  The claim is true for $k=\ell$. For $k<\ell$,
  $\sigma_{k,\ell}=\sigma_{k,k+1}\sigma_{k+1,\ell}$. As $b_\ell$ is the $(k+1)$-th basis vector of
  $\sigma_{k+1,\ell}B$, with the above identity we get
  $\pot(\sigma_{k,\ell}B)=\pot(\sigma_{k,k+1}\sigma_{k+1,\ell}B)=\frac{\|\pi_k(b_{\ell})\|^2}{\|\pi_k(b_k)\|^2}
  \cdot \pot(\sigma_{k+1,\ell}B)$, which completes the proof.
\end{proof}

\section{The Potential-LLL Reduction}
\label{sec:algorithm}

In this section we present our polynomial time variant of \dlll. We start
with the definition of a $\delta$-\plll\ reduced basis. Then we present an
algorithm that outputs such a basis followed by a runtime proof. 
\begin{definition}
\label{def:potLLL}
Let $\delta \in (1/4,1]$. A lattice basis $B=[b_1,\dots,b_n]$ is \emph{$\delta$-\plll reduced} if and only if
\begin{enumerate}
\item $\forall 1\leq j<i\leq n: |\mu_{i,j}| \leq \frac{1}{2}$ (size-reduced).
\item $\forall 1\leq k<\ell\leq n: \delta \cdot \pot(B) \leq
\pot(\sigma_{k,\ell}B)$.
\end{enumerate}

\end{definition}

\begin{lemma}
A $\delta$-\plll reduced basis $B$ is also $\delta$-LLL reduced. 
\end{lemma}
\begin{proof}
  Lemma~\ref{lem:pot} shows that $\delta \cdot \pot(B) \leq \pot(\sigma_{i,i+1}B)$ if and only if
  $\delta \|\pi_i(b_i)\|^2 \leq\|\pi_i(b_{i+1})\|^2$. Thus the Lov\'asz condition is implied by the
  second condition in Definition~\ref{def:potLLL} restricted to consecutive pairs, i.e.\ $\ell = k +
  1$.
\end{proof}

\begin{lemma}
\label{lem:dlllplll}
For $\delta \in \bigl(4^{-1/(n-1)},1\bigr]$, a $\delta$-\dlll reduced basis
$B$ is also $\delta^{n-1}$-\plll reduced.
\end{lemma}
\begin{proof}
  We proceed by contradiction.  Assume that $B$ is not $\delta^{n-1}$-\plll reduced, i.e. there
  exist $1\leq k < \ell \leq n$ such that $\delta^{n-1}\pot(B)> \pot(\sigma_{k,\ell}B)$. By
  Lemma~\ref{lem:pot} this is equivalent to
  \[
  \delta^{n-1} > \prod^\ell_{i=k}
  \frac{\|\pi_i(b_\ell)\|^2}{\|\pi_i(b_i)\|^2}=\prod^{\ell-1}_{i=k}
  \frac{\|\pi_i(b_\ell)\|^2}{\|\pi_i(b_i)\|^2}\,.
  \]
  It follows that there exist a $j \in [k, \ell-1]$ such that
  $\|\pi_j(b_\ell)\|^2/\|\pi_j(b_j)\|^2<\delta^{(n-1)/(\ell-k)}\leq\delta$ which implies that $B$ is
  not $\delta$-\dlll reduced.
\end{proof}

\subsection{High-Level Description}
\label{S:hldesc}

A high-level version of the algorithm is presented as Algorithm~\ref{alg:potLLL}. The algorithm is
very similar to the classical LLL algorithm and the classical \dlll reduction by Schnorr and Euchner
\cite{BKZ}. During its execution, the first $\ell-1$~basis vectors are always $\delta$-\plll
reduced (this guarantees termination of the algorithm). As opposed to classical
LLL, and similar to \dlll, $\ell$ might decrease by more than
one. This happens precisely during deep insertions: in these cases, the $\ell$-th vector is not
swapped with the $(\ell-1)$-th one, as in classical LLL, but with the $k$-th one for $k < \ell -
1$. In case $k = \ell - 1$, this equals the swapping of adjacent basis vectors as in classical LLL.
The main difference of \plll\ and \dlll\ is the condition that controls
insertion of a vector.

\paragraph{Preprocessing} On line~\ref{alg:potLLL:pp} we LLL reduce the input basis before
proceeding. It turns out that while omitting this preprocessing does not change the output quality
of the bases on average (see Figure~\ref{fig:approx-pp-vs-up}), it is on average beneficial when it
comes to the running time (see Figure~\ref{fig:time-pp-vs-up}). Note that most implementations of
BKZ also preprocess the input basis with LLL. The figures can be found on
page~\pageref{fig:pp-vs-up}, and a more detailed discussion in Section~\ref{S:exp:preproc}.

\paragraph{\pllltwo} On line~\ref{alg:potLLL:min} the insertion depth is chosen such that the
potential of the basis is minimal under the insertion. Alternatively one can choose the insertion
place $k$ as $\min\{k: \pot(\sigma_{k,\ell}B)<\delta \cdot \pot(B)\}$.  Neither the running time
analysis nor the fact that the output basis is \plll reduced is changed. We refer to this variant of
\plll as \pllltwo.

\begin{algorithm}[htbp]
  \caption{Potential LLL}
  \label{alg:potLLL}
  \SetKwComment{Comment}{$\vartriangleright$~}{}
  \SetCommentSty{textit}
  \DontPrintSemicolon
  
  \KwIn{Basis $B \in \Z^{n \times m}$, $\delta \in (1/4,1]$}
  \KwOut{A $\delta$-\plll reduced basis.}
  $\delta$-LLL reduce $B$\; \label{alg:potLLL:pp}
  $\ell \leftarrow 1$\;
  \While{$\ell \leq n$}{
    	Size-reduce$(B)$\;\label{alg:potLLL:sizereduce}
	$k \leftarrow \argmin_{1\leq j\leq \ell}\pot(\sigma_{j,\ell}B)$\label{alg:potLLL:min}\;	
	\eIf{$\delta \cdot \pot(B)>\pot(\sigma_{k,\ell}B$)\label{alg:potLLL:if}}{
			$B \leftarrow \sigma_{k,\ell}B$\;
			$\ell \leftarrow k$ \;
	}{
		$\ell \leftarrow \ell+1$ \;
	}
  }
\Return $B$\;
\end{algorithm}

\subsection{Detailed Description}
There are two details to consider when implementing Algorithm~\ref{alg:potLLL}. The first one is
that since the basis vectors~$b_1, \dots, b_{\ell-1}$ are already $\delta$-\plll reduced, they are
in particular also size-reduced. Moreover, the basis vectors $b_{\ell+1}, \dots, b_n$ will be
considered later again. So in line~\ref{alg:potLLL:sizereduce} of the algorithm
it suffices to
size-reduce $b_\ell$ by $b_1, \dots, b_{\ell-1}$ as in classical LLL. Upon termination, when $\ell =
n + 1$, the whole basis will be size-reduced.

Another thing to consider is the computation of the potentials of $B$ and $\sigma_{j,\ell} B$ for $1
\le j \le \ell$ in line~\ref{alg:potLLL:min}.  Computing the potential of the basis is a rather slow
operation. But we do not need to compute the potential itself, but only compare
$\pot(\sigma_{k,\ell} B)$ to $\pot(B)$; by Lemma~\ref{lem:pot}, this quotient can be efficiently
computed. Define $P_{k,\ell} := \pot(\sigma_{k,\ell} B) / \pot(B)$. The ``if''-condition in
line~\ref{alg:potLLL:if} will then change to $\delta > P_{k,\ell}$, and the minimum in
line~\ref{alg:potLLL:min} will change to $\argmin_{1 \leq j \leq \ell} P_{j,\ell}$. Using
$P_{\ell,\ell} = 1$ and
\begin{equation}
\label{equ:P_jl}
P_{j,\ell}= \frac{\pot(\sigma_{j,\ell}
B)}{\pot(B)}=P_{j+1,\ell} \cdot \frac{\|\pi_{j}(b_\ell)\|^2}{\|\pi_{j}(b_{j})\|^2} =
P_{j+1,\ell} \cdot \frac{\|b_\ell^*\|^2+ \sum_{i=j}^{\ell-1} \mu_{\ell,i}^2
\|b_i^*\|^2}{\|b_{j}^*\|^2}
\end{equation}
for $j < \ell$ (Lemma~\ref{lem:pot}), we can quickly determine $\argmin_{1 \leq j \leq \ell}
P_{j,\ell}$ and check whether $\delta > P_{k,\ell}$ if $j$ minimizes $P_{j,\ell}$.

 A detailed version of Algorithm~\ref{alg:potLLL} with these steps
filled in is described as Algorithm~\ref{alg:potLLL:full}.
On line~\ref{alg:potLLL:full:potmod} of Algorithm~\ref{alg:potLLL:full}, $P_{j,\ell}$ is iteratively computed as in
Equation~\eqref{equ:P_jl}.  Clearly, the algorithm could be further improved by iteratively
computing $\|\pi_j(b_\ell)\|^2$ from $\|\pi_{j+1}(b_\ell)\|^2$.  Depending on the implementation of
the Gram-Schmidt orthogonalization, this might already have been computed and stored. For example,
when using the Gram-Schmidt orthogonalization as described in Figure~4
of~\cite{nguyen-stehle-fplll-revisited}, then $\|\pi_j(b_\ell)\|^2 = s_{j-1}$ after computation of
$\|b_\ell^*\|^2$ and $\mu_{\ell,j}$ for $1 \le j < \ell$.

\begin{algorithm}[htbp]
  \caption{Potential LLL, detailed version}
  \label{alg:potLLL:full}
  \SetKwComment{Comment}{$\vartriangleright$~}{}
  \SetCommentSty{textit}
  \DontPrintSemicolon
  
  \KwIn{Basis $B \in \Z^{n \times m}$, $\delta \in (1/4,1]$}
  \KwOut{A $\delta$-\plll reduced basis.}
  $0.99$-LLL reduce $B$\; \label{alg:potLLL:full:pp}
  $\ell \leftarrow 1$\;
  \While{$\ell \leq n$\label{alg:potLLL:full:while}}{
        Size-reduce$(b_\ell \text{ by } b_1, \dots, b_{\ell-1})$\;\label{alg:potLLL:full:sizereduce}
        Update($\|b_\ell^*\|^2$ and $\mu_{\ell,j}$ for $1 \le j <
\ell$)\;\label{alg:potLLL:full:update1}
        $P \leftarrow 1$, \quad $P_{\min} \leftarrow 1$, \quad $k \leftarrow 1$ \;
        \For{$j=\ell-1$ down to $1$}{
                $P \leftarrow P \cdot \frac{ \|b_\ell^*\|^2 +\sum_{i=j}^{\ell-1} \mu_{\ell,i}^2 \|b_i^*\|^2}{\|b_j^*\|^2}$\label{alg:potLLL:full:potmod} \;
                \If{$P < P_{\min}$}{
                        $k \leftarrow j$ \;
                        $P_{\min} \leftarrow P$ \;
                }
        }
        \eIf{$\delta > P_{\min}$\label{alg:potLLL:full:if}}{
                $B \leftarrow \sigma_{k,\ell}B$\;
                Update($\|b_k^*\|^2$ and $\mu_{k,j}$ for $1 \le j <
k$)\;\label{alg:potLLL:full:update2}
                $\ell \leftarrow k$ \;
        }{
                $\ell \leftarrow \ell+1$ \;
        }
  }
\Return $B$\;
\end{algorithm}

\subsection{Complexity Analysis}

Here we show that the number of operations in the \plll algorithm are
bounded polynomially in the dimension $n$ and the logarithm of the input size.
We present the runtime for Algorithm~\ref{alg:potLLL:full}.

\begin{proposition}
\label{prop:potLLL2}
Let $\delta \in (1/4, 1)$ and $C = \max_{i=1\ldots n} \norm{b_i}^2$. Then
Algorithm~\ref{alg:potLLL:full} performs $\grossO{n^3 \log_{1/\delta}(C)}$ iterations of the
\texttt{while} loop in line~\ref{alg:potLLL:full:while} and a total of $\grossO{m n^4
  \log_{1/\delta}(C)}$ arithmetic operations.
\end{proposition}
\begin{proof}
Let us start by upper bounding the potential $I$ of the input basis with respect to $C$. 
Let
$d_j := \vol\left(\tL(b_1,\dots,b_j)\right)^2 = \prod_{i=1}^j \norm{b_i^*}^2$ for $j=1,\dots,n$.
Recall that $\norm{b_i^*}^2 \leq \norm{b_i}^2 \leq C$ for $i=1,\dots,n$ and hence $d_j<C^j$. Consequently we have the following upper bound on the potential
\begin{equation}
\label{equ:I}
I = \prod_{j=1}^{n-1} d_j \cdot \vol(\tL) \leq \prod_{j=1}^{n-1} C^j \cdot \vol(\tL) \leq
C^{\frac{n(n-1)}{2}} \cdot \vol(\tL)\,.
\end{equation}
Now, by a standard argument, we show that the number of iterations of the
while loop is bounded by $\grossO{n^3 \log_{1/\delta}(C)}$.
In each iteration, either the iteration counter~$\ell$ is increased by 1, or an
insertion takes place
and $\ell$ is decreased by at most $n - 1$. In the
insertion case, the potential is decreased by a factor at least $\delta$. So
after $N$ swaps the potential $I_N$ satisfies $I \geq (1/\delta)^NI_N \geq
(1/\delta)^N \cdot \vol(\tL)$ using the fact that $I_N \geq \vol(\tL)$.
Consequently the number of swaps~$N$ is
bounded by $N \leq \log_{1/\delta} (I / \vol(\tL))$.
By Equation (\ref{equ:I}) we get that $N \leq \log_{1/\delta}\bigl(C^{n(n-1)/2}\bigr)$.
Now note that
the number~$M$ of iterations where~$\ell$ is increased by 1 is at most $M \le (n - 1) \cdot N + n$.
This shows that the number of iterations is bounded by $\grossO{n^3
\log_{1/\delta}(C)}$.

Next we show that the number of operations performed in each iteration of the \texttt{while} loop is
dominated by $\grossO{n m}$ operations.
Size-reduction (line~\ref{alg:potLLL:full:sizereduce}) and the first update
step (line~\ref{alg:potLLL:full:update1}) can be done in \grossO{n m} steps. The
for-loop consists of \grossO{n} iterations where the most expensive operation
is the update of $P$ in line~\ref{alg:potLLL:full:potmod}. Therefore the loop
requires \grossO{n m} arithmetic operations. Insertion can be done in \grossO{n}
operations, whereas the second update in line~\ref{alg:potLLL:full:update2}
requires again \grossO{n m} operations.

It follows that each iteration costs at most \grossO{n m} arithmetic operations. This shows
that in total the algorithm performs $\grossO{m n^4 \log(C)}$ operations.
\end{proof}

\subsection{Worst-Case Behavior}
\label{sec:critical}
For $\delta=1$, there exist so called \emph{critical bases} which are
$\delta$-LLL reduced bases and whose Hermite factor reaches the worst case bound
in~(\ref{equ:hermiteLLL})~\cite{sc94}.
These bases can be adapted to form a \dlll reduced basis where the first vector
reaches the worst case bound  in~(\ref{equ:hermiteLLL}).

\begin{proposition}\label{prop:critical}
For $\alpha = \sqrt{3/4}$,  the rows of $B = A_n(\alpha)$ (see below) define a
$\delta$-\dlll reduced basis with $\delta=1$ and
$\|b_1\|^2=\frac{1}{\alpha^{(n-1)/2}} \vol(\tL(A_n))^{1/n}$.
\end{proposition}

\begin{equation}
A_n(\alpha):=\left(
\begin{matrix}
1			&	0				& 	 \cdots			& \cdots 		&		\cdots			&0\\
\frac{1}{2}	& 	\;\alpha 			&	\ddots
&			&						&
\vdots\\
\vdots		&	\;\frac{\alpha}{2}	& \;\alpha^2
& \ddots		&						& \vdots \\
 \vdots 	&	\vdots 			& \frac{\alpha^2}{2}	&\ddots		&	\ddots				&  \vdots \\
 \vdots		&	\vdots			&\vdots 				&\ddots		&\alpha^{n-2}				& 0 \\
\frac{1}{2}	&	\;\frac{\alpha}{2}	&\;\frac{\alpha^2}{2}	& \hdots
		& \;\frac{\alpha^{n-2}}{2}	& \;\alpha^{n-1}
\end{matrix}
\right)
\end{equation}

\begin{proof}
From the diagonal form of $A_n$ it is easy to see that $\vol(\tL)=\det(A_n)=\alpha^{n(n-1)/2}$.
Hence $\|b_1\|^2=1=1/\alpha^{(n-1)/2}\vol(\tL)$. 
It remains to show that $A_n$ is \dlll reduced. Note that the
orthogonalized basis $B^*$ is a diagonal
matrix with the same entries on the diagonal as $B$.
Note that it is size reduced as for all $1\leq j<i\leq n$ we have
$\mu_{i,j}=\langle b_i,b^*_j\rangle/\langle b^*_j,b^*_j
\rangle=\tfrac{1}{2} \alpha^{2 (j-1)}/\alpha^{2 (j-1)}=\frac{1}{2}$.
Further, using that $\pi_j(b_i)=b_i^* +
\sum_{\ell=j}^{i-1}\mu_{i,\ell}b_{\ell}^*$, we have that
\[
\|\pi_j(b_i)\|^2=\alpha^{2(i-1)}+\frac{1}{4}\sum^{i-1}_{\ell=j}
\alpha^{2(\ell-1)}=\alpha^{2(j-1)}\left(\frac{1}{4}\sum^{i-j-1}_{\ell=0}
\alpha^{2\ell}+\alpha^{2(i-j)}\right)\;.
\]
As for $\alpha=\sqrt{3/4}$, we have that
$\frac{1}{4}\sum^{i-j-1}_{\ell=0} \alpha^{2\ell}+\alpha^{2(i-j)}=1$,
and hence $\|\pi_j(b_i)\|^2=\alpha^{2(j-1)}= \|\pi_j(b_j)\|^2$. Therefore, the
norms of the projections for fixed~$j$ are all equal, and $A_n(\alpha)$ is $\delta$-\plll reduced with
$\delta=1$.
\end{proof}
Using Lemma~\ref{lem:dlllplll}, we obtain:
\begin{corollary}
For $\alpha = \sqrt{3/4}$,  the rows of $A_n(\alpha)$ define a $\delta$-\plll reduced basis
with $\delta=1$ and $\|b_1\|^2=\frac{1}{\alpha^{(n-1)/2}} \vol(\tL(A_n))^{1/n}$. \qed
\end{corollary}

\section{Experimental Results}\label{sec:experiments}
Extensive experiments have been made to examine how the classical LLL reduction algorithm performs
in practice \cite{ng06,gam08}. We ran extensive experiments to compare our \plll 
algorithms to our implementations of LLL, \dlll, and BKZ.

\subsection{Setting}
We run the following algorithms, each with the standard reduction parameter $\delta=0.99$:
\begin{enumerate}
\item classical LLL,
\item \plll and \pllltwo,
\item \dlll with blocksize $\beta=5$ and $\beta=10$,
\item BKZ with blocksize $5$ (BKZ-5) and $10$ (BKZ-10).
\end{enumerate}
The implementations all use the same arithmetic back-end. Integer arithmetic is done using GMP, and
Gram-Schmidt arithmetic is done as described in~\cite[Figures~4 and
5]{nguyen-stehle-fplll-revisited}. As floating point types, \texttt{long double} (x64 extended
precision format, 80~bit representation) and MPFR arbitrary precision floating point numbers are
used with a precision as described in \cite{nguyen-stehle-fplll-revisited}. The implementations of
\dlll and BKZ follow the classical description in \cite{BKZ}. \plll was implemented as described in
Algorithm~\ref{alg:potLLL:full} (page~\pageref{alg:potLLL:full}). Our implementation will be made
publicly available.

We ran experiments in dimensions 40 to 400, considering the dimensions which are multiples of~10.
Some algorithms become too slow in high dimensions, whence we restrict the dimensions for these as
follows: For \dlll with $\beta = 10$ we ran experiments up to dimension~300 and for \pllltwo and
BKZ-10 up to dimension 350.

In each dimension,
we considered 50 random lattices. More precisely, we used the HNF bases of the lattices of seed 0 to 49 from the SVP
Challenge.\footnote{\url{http://www.latticechallenge.org/svp-challenge}}

All experiments were run on Intel\textsuperscript{\textregistered}
Xeon\textsuperscript{\textregistered} X7550 CPUs at 2~GHz on a shared memory machine. For dimensions
40 up to 160, we used \texttt{long double} arithmetic, and for dimensions above 160, we used
MPFR. In dimension 160, we did the experiments both using \texttt{long double} and MPFR
arithmetic. The reduced lattices did not differ. In dimension~170, floating point errors prevented
the \texttt{long double} arithmetic variant to complete on some of the lattices.

\subsection{Preprocessing}
\label{S:exp:preproc}
As mentioned in Section~\ref{S:hldesc}, we added a ``preprocessing'' step to \plll, \pllltwo and
\dlll, by first running LLL without any deep insertions and with the same reduction parameter on the
basis, and only then running \plll resp.\ \dlll. We performed all experiments both with and without
this preprocessing, except that without preprocessing, we left out certain higher dimensions. More
precisely, \plll was run until dimension~400, \plll was run until dimension~300, \dlll with $\beta =
5$ up to dimension~320, and \dlll with $\beta = 10$ up to dimension~250.

Figure~\ref{fig:approx-pp-vs-up} (see page~\pageref{fig:pp-vs-up}) shows the average $n$-th root
Hermite factor for the resulting bases. It appears that while preprocessing can have both a positive
and negative impact on the output quality, it in general does not change the average $n$-th root
Hermite factor. This was to be expected, since essentially we applied \plll resp.\ \dlll to two
different bases of the same lattice: one in Hermite Normal Form, and the other 0.99-LLL reduced.

When comparing the timing results, on the other hand, there are large
differences. Figure~\ref{fig:approx-pp-vs-up} shows the timing in dimensions~160 up to 400 for \dlll
and \plll with and without preprocessing. The times for the algorithms with preprocessing include
the time needed for applying LLL with $\alpha = 0.99$. It is clear that the algorithms with
preprocessing are significantly faster than the ones without.

We conclude that while preprocessing does not change the output quality in average, it has a huge
impact on the running time. For this reason, and also to have a better comparison to BKZ which
always applies LLL first, we restricted to the algorithms with preprocessing for the rest of the
experiments.

\subsection{Results}

For each run, we recorded the length of the shortest vector as well as the required CPU time for the
reduction. Our main interest lies in the $n$-th root of the \emph{Hermite factor}
$\frac{\|b_1\|}{\vol(\tL)^{1/n}}$, where $b_1$ is the shortest vector of the basis of $\tL$
returned.

Figure~\ref{fig:overview:apfa} (see page~\pageref{fig:overview:apfa}) compares the average $n$-th
root Hermite factor achieved by the different reduction algorithms in all dimensions.  Also
indicated are the confidence intervals for the average value with a confidence level of 99.9\%. The
average values for dimensions 100, 200, 300 and 400 are additionally summarized in
Table~\ref{table:red}, where also values of the worst-case bound from
Equation~\eqref{equ:hermiteLLL} are given.  Note that our data for LLL is similar to the one in
\cite{ng06} and \cite[Table~1]{gam08}. However, we do not see convergence of the $n$-th root Hermite
factors in our experiments, as they are still increasing in high dimensions $n>200$, respectively
even slightly decreasing in the case of LLL.  Our \plll algorithm clearly outperforms LLL and BKZ-5,
however not \pllltwo, \dlll with $\beta=5,10$ and BKZ-10.  \dlll with $\beta=10$ seems the strongest
of the considered lattice reduction algorithms.  It is very interesting to see that \pllltwo
performs remarkably better than the original \plll when it comes to the Hermite factor achieved.

Figure~\ref{fig:overview:time} (see page~\pageref{fig:overview:time}) compares the average
logarithmic running time of the algorithms for all dimensions. Recall that we used different
arithmetic for dimensions below and above 160, whence two separate graphs are given.  We see that
the observed order is similar to the order induced by the Hermite factors. The only somewhat
surprising fact is that \pllltwo is even slower than BKZ-10, i.e.\ it is only faster than \dlll with
$\beta=10$.

\begin{table}[h]
  \renewcommand{\arraystretch}{1.3}
  \begin{center}
    \begin{tabular}{  l l l l l l l }
      \toprule
      Dimension & & $n = 100$ & $n = 200$ & $n = 300$ & $n = 400$ \\
      \midrule
      Worst-case bound (proven) & ~\ & $\approx 1.0774$ & $\approx 1.0778$ & $\approx 1.0779$ & $\approx 1.0780$ \\
      Empirical $0.99$-LLL & & $ 1.0186$ & $ 1.0204$ & $ 1.0212$ & $ 1.0212$ \\ 
      Empirical $0.99$-BKZ-5 & ~\ & $ 1.0152$ & $ 1.0160$ & $ 1.0162$ & $ 1.0164$ \\ 
      Empirical $0.99$-\plll & & $ 1.0146$ & $ 1.0151$ & $ 1.0153$ & $ 1.0155$ \\ 
      Empirical $0.99$-\pllltwo & & $ 1.0142$ & $ 1.0147$ & $ 1.0149$ & \;\quad--- \\ 
      Empirical $0.99$-\dlll with $\beta=5$ & ~\ & $ 1.0137$ & $ 1.0146$ & $ 1.0150$ & $ 1.0152$ \\ 
      Empirical $0.99$-BKZ-10 & ~\ & $ 1.0139$ & $ 1.0144$ & $ 1.0145$ & \;\quad--- \\ 
      Empirical $0.99$-\dlll with $\beta=10$ & ~\ & $ 1.0129$ & $ 1.0134$ & $ 1.0138$ & \;\quad--- \\ 
      \toprule
    \end{tabular}
  \end{center}
  \caption{Worst case bound and average case estimate for $\delta$-LLL reduction, $\delta$-\dlll
  reduction, $\delta$-\plll reduction and $\delta$-BKZ reduction of the $n$-th root Hermite factor
  $\|b_1\|^{1/n} \cdot \vol(\tL)^{-1/n^2}$. The entries are sorted in descending order with respect
  to the observed Hermite factors.}
  \label{table:red}
\end{table}

Figures~\ref{fig:comparism:up:ld} and \ref{fig:comparism:up:real} (see
pages~\pageref{fig:comparism:up:ld} and \pageref{fig:comparism:up:real}) allow to compare the
different reduction algorithms with respect to the running time and the achieved Hermite factor at
the same time. Every line connecting bullets corresponds to the behavior of one algorithm in
different dimensions. Again, the gray box surrounding a bullet is a two-dimensional confidence
interval with confidence level 99.9\%. The shaded regions show which Hermite factors can be achieved
in every dimension by these algorithms. Algorithms on the border of the region are optimal for their
Hermite factor: none of the other algorithms in this list produces a better average Hermite factor
in less time.

The only algorithm which is never optimal is \pllltwo, which is slower than \dlll
with $\beta = 5$ or
BKZ-10 and provides worse average Hermite factors up to dimension 160. \pllltwo produces slightly
better average Hermite factors than \dlll with $\beta = 5$ in high dimensions, for example from 280
on, but is there beaten by BKZ-10 which is in these dimensions far more efficient and provides
better Hermite factors.

Another interesting observation is that in dimensions 40 to 80, \plll is both faster than BKZ-5 and
yields shorter vectors. While the running time difference in dimension~80 is quite marginal, it is
substantial in dimension~40. This shows that \plll could be used for efficient preprocessing of
blocks for enumeration in BKZ-style algorithms with large block sizes, such as Chen's and Nguyen's
BKZ~2.0 \cite{bkz-2.0}.

\subsection{Comparison to fplll}
To show the independence of the \plll concept from the concrete implementation, we added a \plll
implementation to version~4.0.1 of the fplll
library;\footnote{\url{http://perso.ens-lyon.fr/damien.stehle/fplll/} and
\url{http://xpujol.net/fplll/}} a patch can be downloaded at
\url{http://user.math.uzh.ch/fontein/fplll-potlll/}. We ran the experiments with fplll's LLL
implementation and our \plll addition in dimensions 40 to 320. For lower dimensions (up to 160 at
least), the fplll-reduced lattices (both LLL and \plll) were identical to the ones of our
implementation. For higher dimensions, the output quality in terms of the $n$-th root Hermite factor
was essentially the same as for our implementation. While fplll was somewhat faster than our
implementation, the relative difference between LLL and \plll was the same as for our
implementation.

\begin{sidewaysfigure}[p]
  \includegraphics[width=\textwidth]{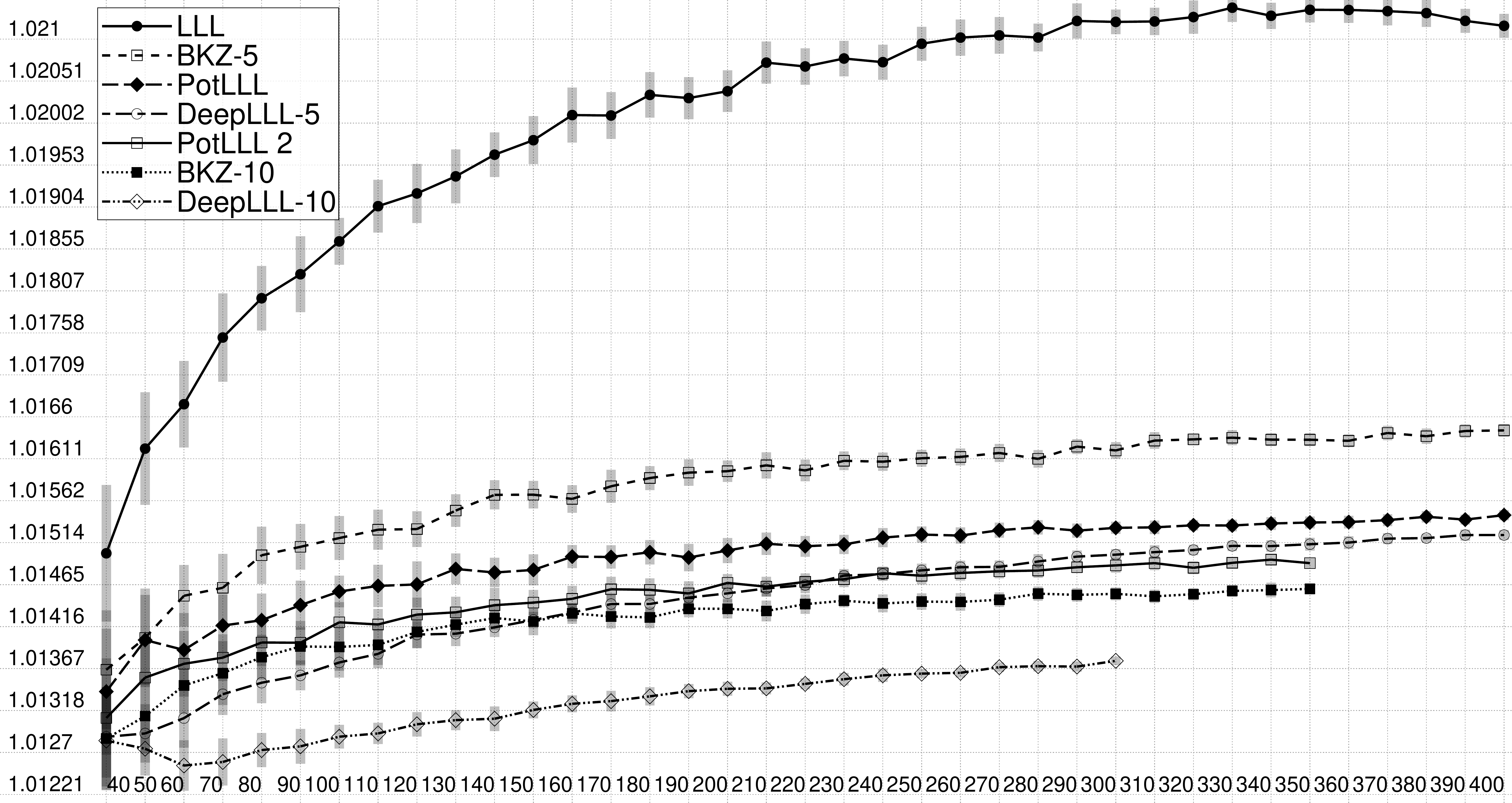}
  \caption{Average $n$-th root Hermite factor ($y$~axis) for dimension~$n$ ($x$~axis) from 40 to
  400.}
  \label{fig:overview:apfa}
\end{sidewaysfigure}

\begin{sidewaysfigure}[p]
  \includegraphics[width=0.3\textwidth]{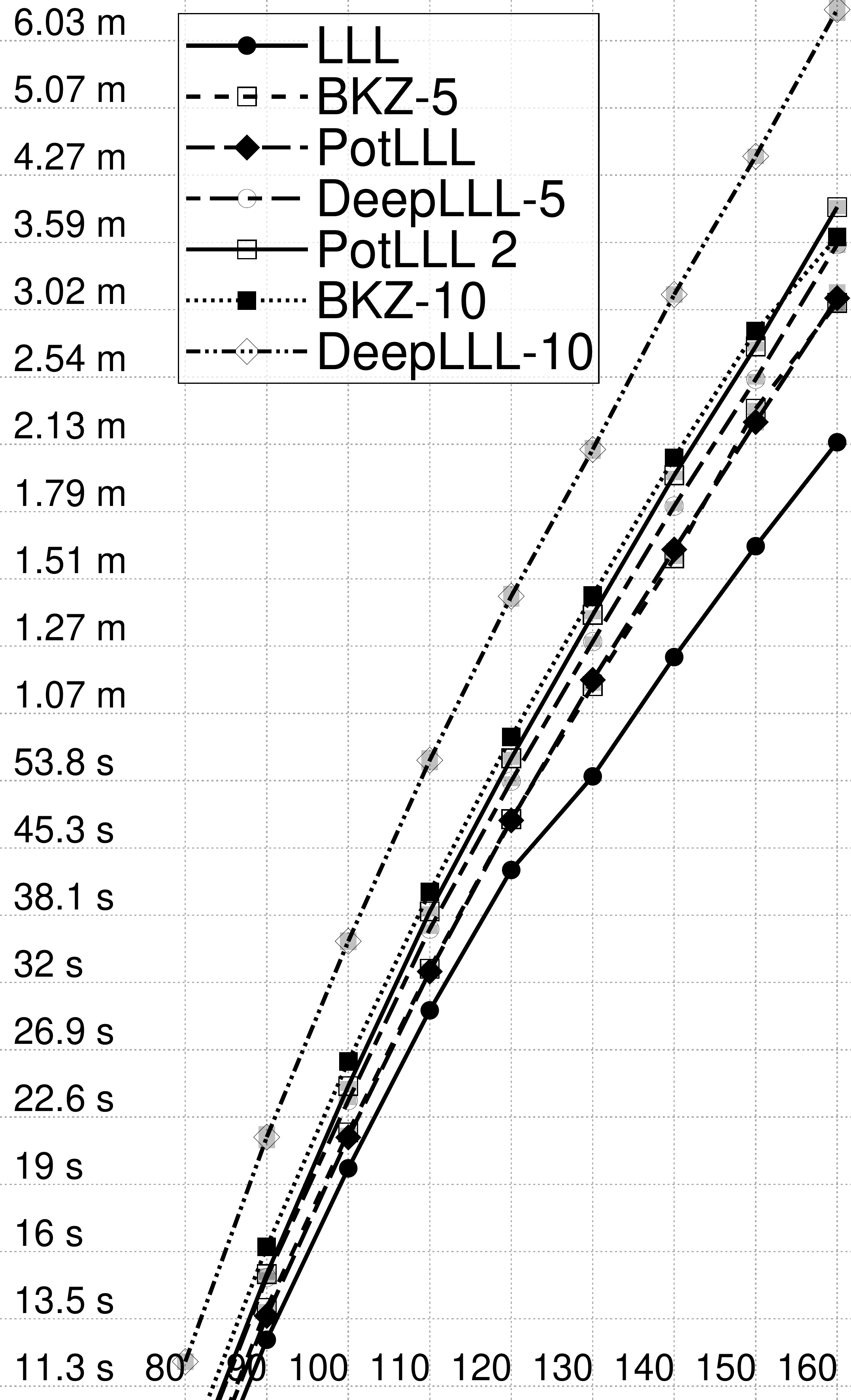}%
  \includegraphics[width=0.7\textwidth]{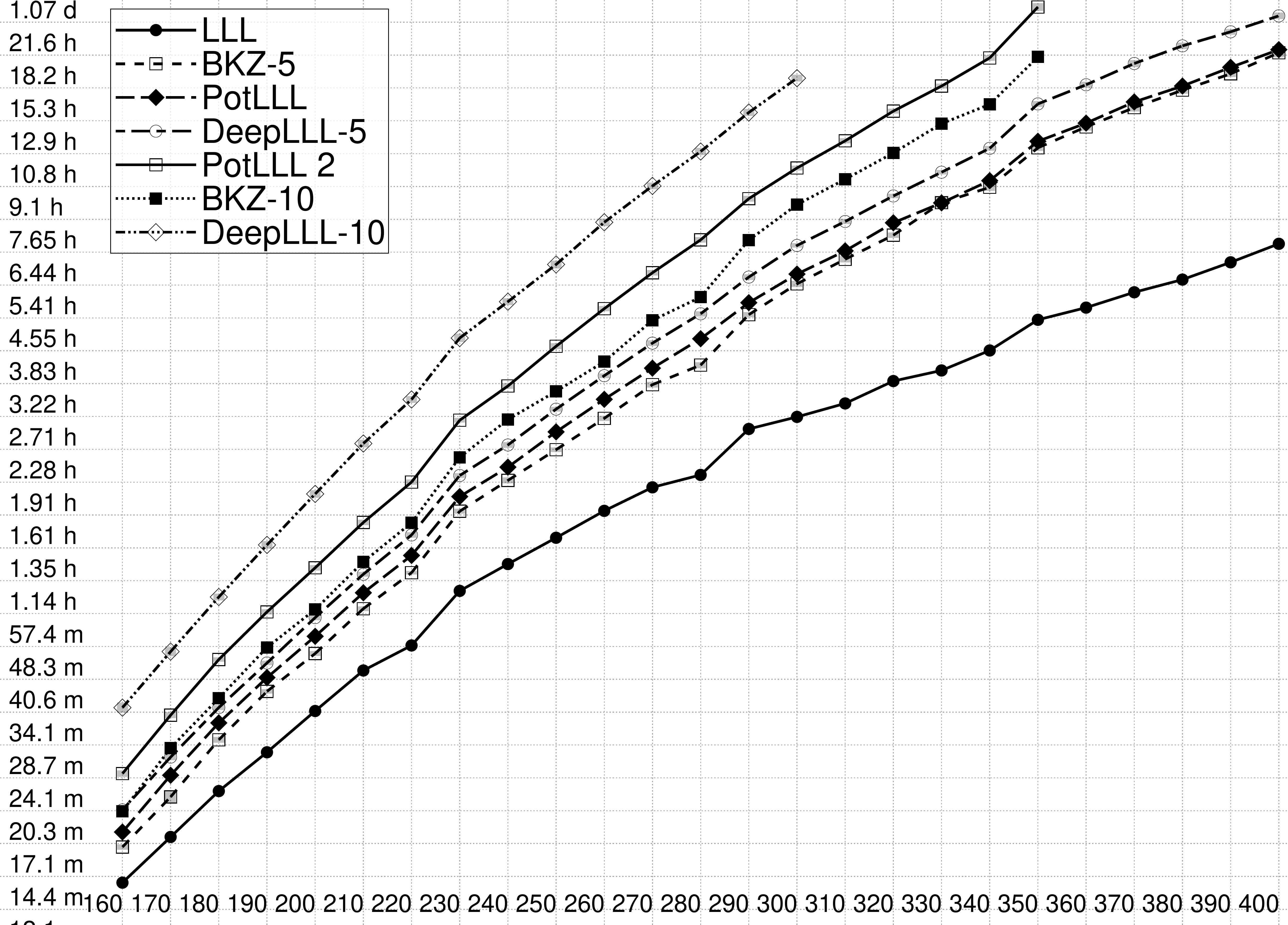}
  \caption{Average logarithmic CPU time ($y$~axis) for dimension~$n$ ($x$~axis) from 40 to 400. The
  left graph uses \texttt{long double} arithmetic, the right graph MPFR arithmetic.}
  \label{fig:overview:time}
\end{sidewaysfigure}

\begin{sidewaysfigure}[p]
  \includegraphics[width=\textwidth]{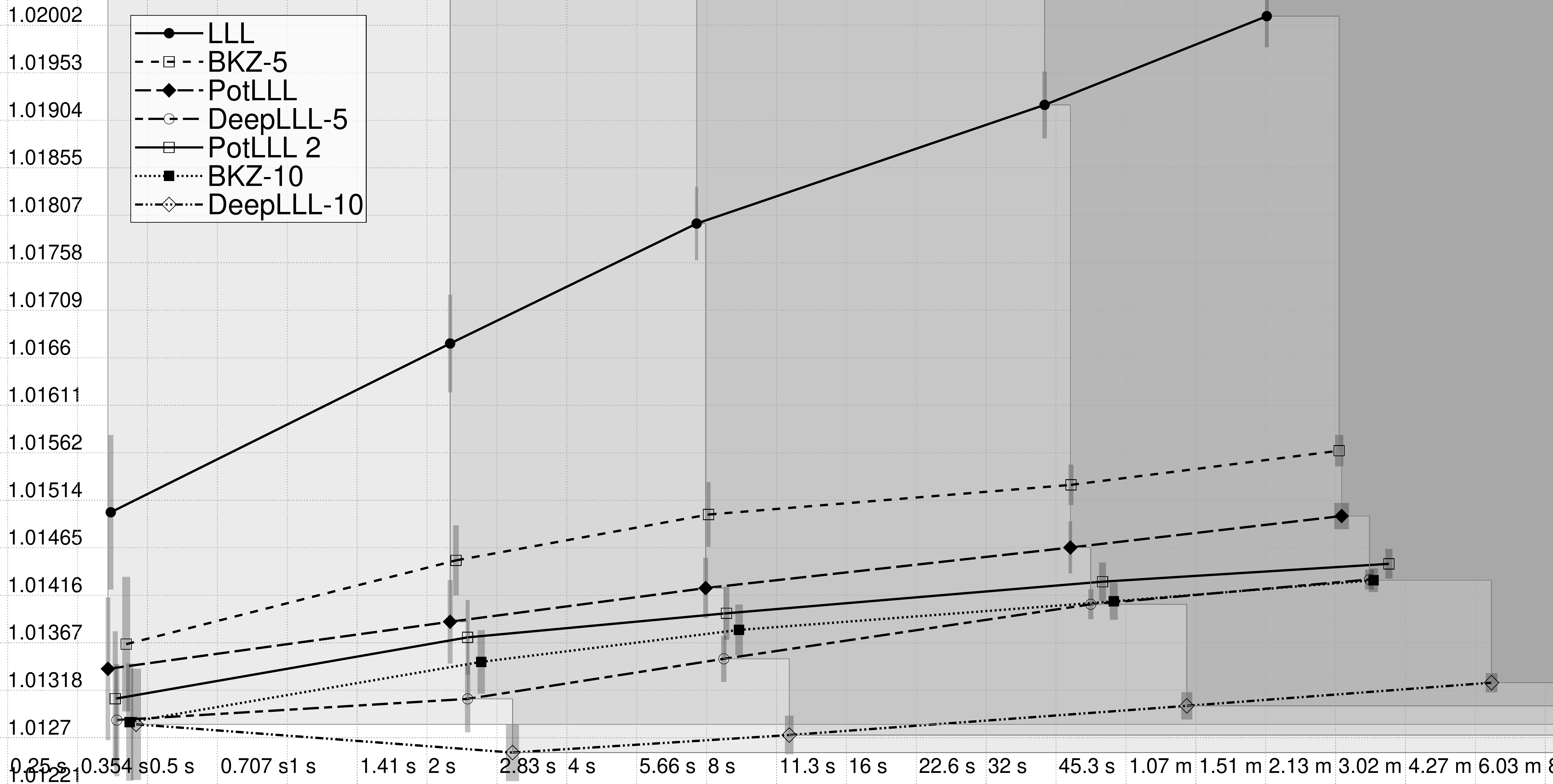}
  \caption{\texttt{long double} arithmetic. The highlighted areas represent dimensions 40, 60, 80,
  120 and 160.}
  \label{fig:comparism:up:ld}
\end{sidewaysfigure}

\begin{sidewaysfigure}[p]
  \includegraphics[width=\textwidth]{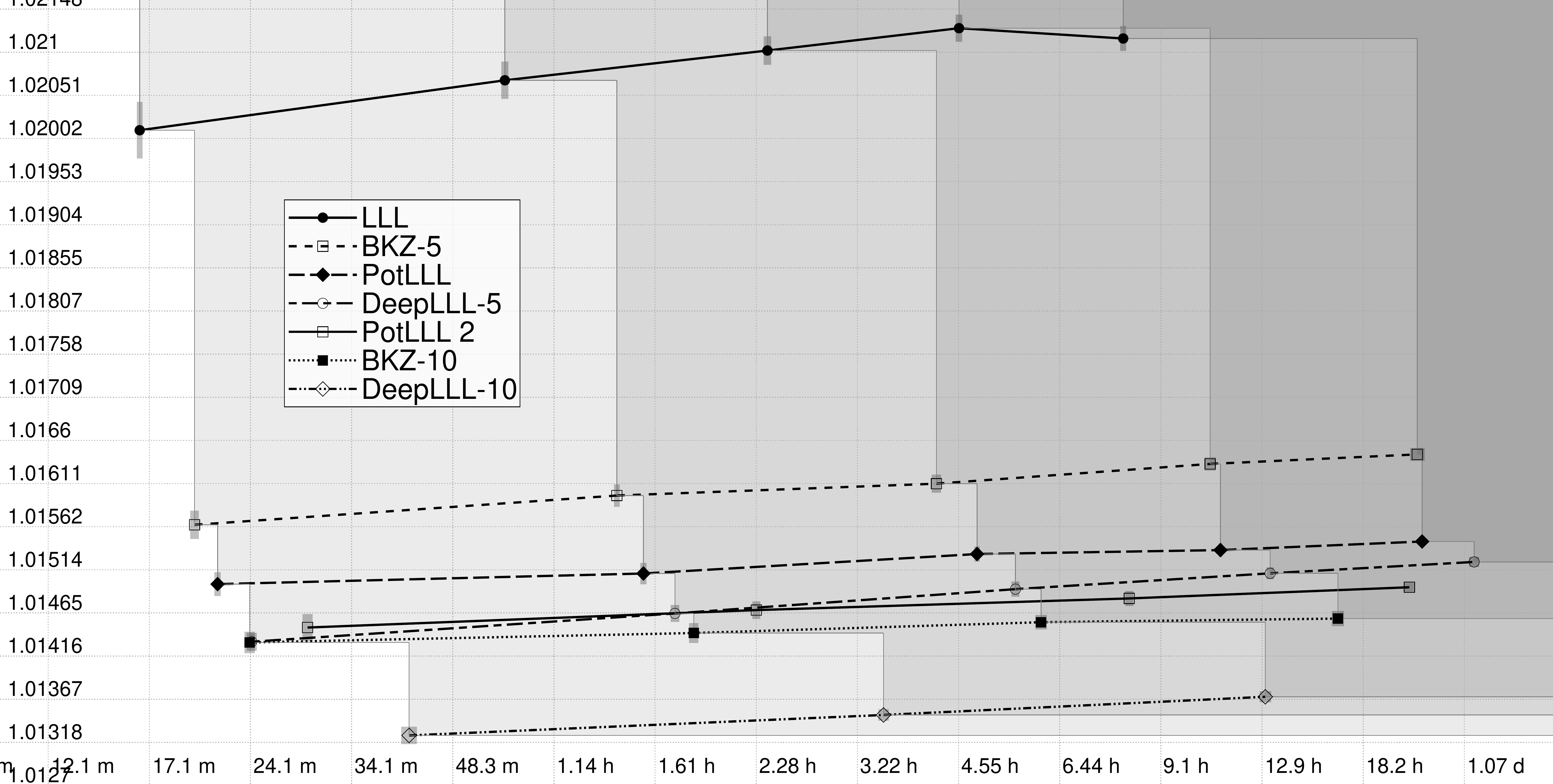}
  \caption{MPFR arithmetic. The highlighted areas represent dimensions 160, 220, 280, 340 and 400.}
  \label{fig:comparism:up:real}
\end{sidewaysfigure}

\begin{figure}[p]
  \begin{center}
    \begin{subfigure}[t]{\textwidth}
      \begin{center}
        \includegraphics[width=\textwidth]{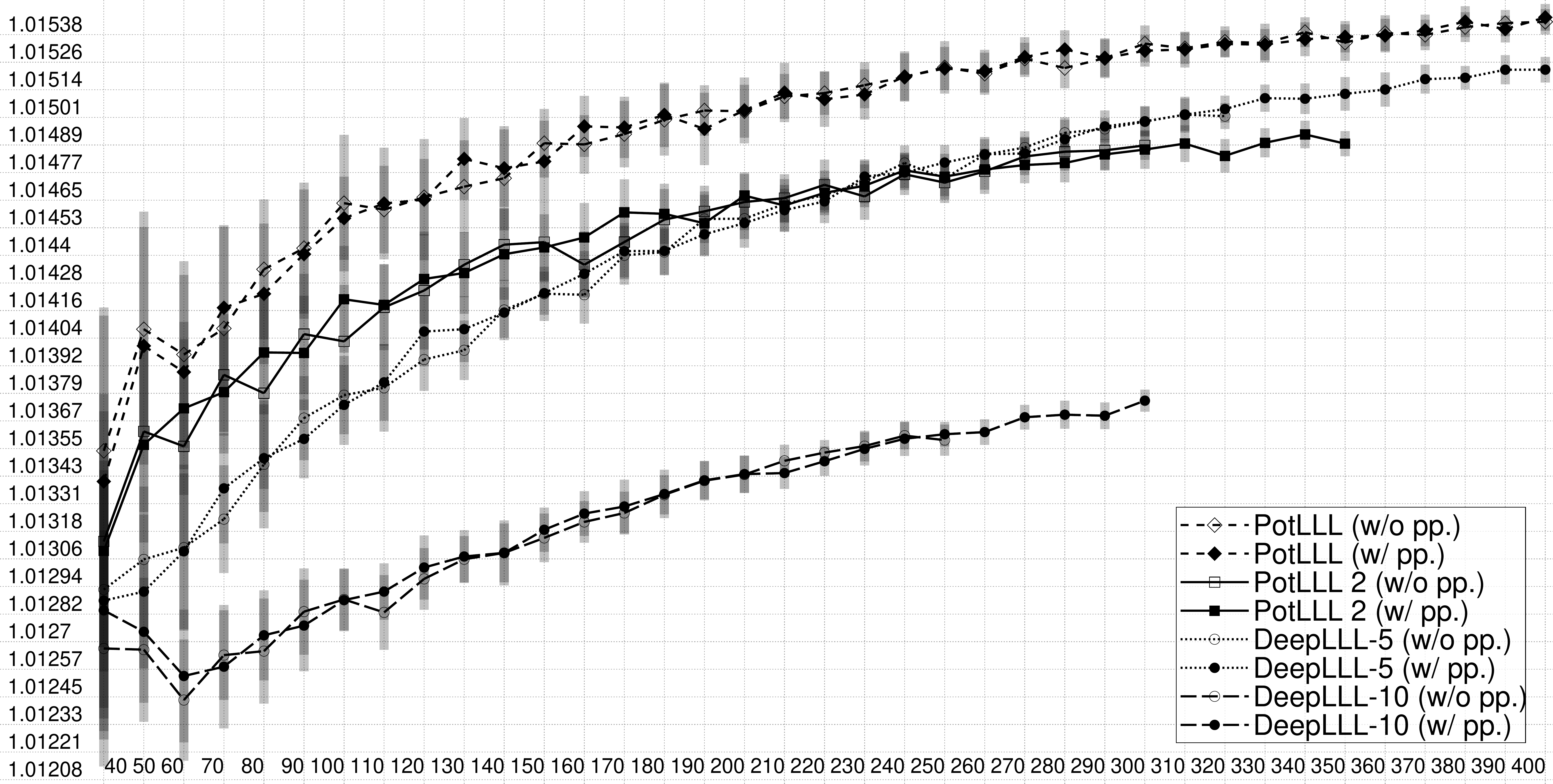}
      \end{center}
      \caption{Comparison of the approximation factors with and without preprocessing.\vspace{0.3cm}}
      \label{fig:approx-pp-vs-up}
    \end{subfigure}
    \begin{subfigure}[t]{\textwidth}
      \begin{center}
        \includegraphics[width=0.5\textwidth]{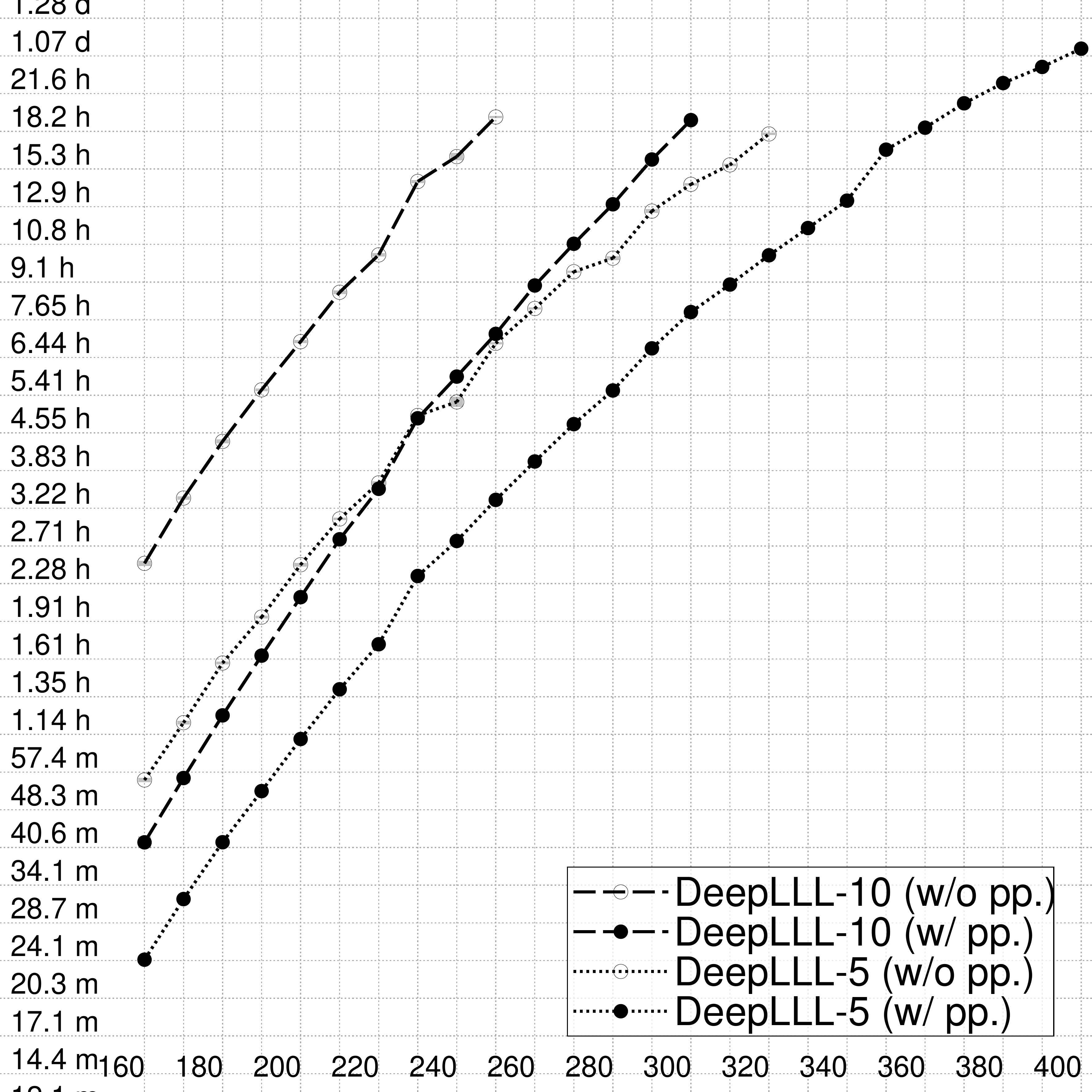}%
        \includegraphics[width=0.5\textwidth]{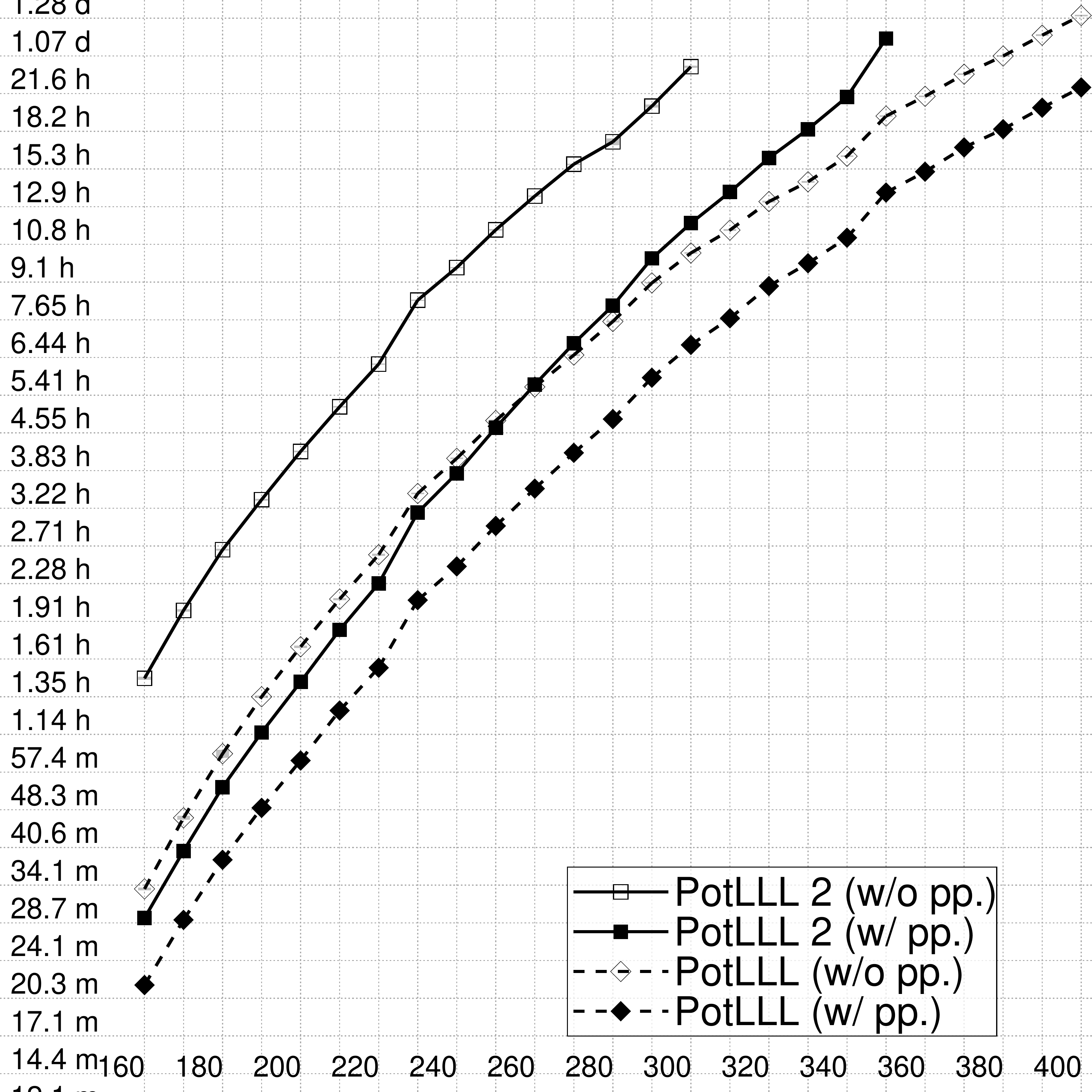}
      \end{center}
      \caption{Comparison of the running times with and without preprocessing (MPFR arithmetic).}
      \label{fig:time-pp-vs-up}
    \end{subfigure}
    \caption{Comparing \plll, \pllltwo and \dlll variants with and without LLL preprocessing.}
    \label{fig:pp-vs-up}
  \end{center}
\end{figure}

\section{Conclusion and future work}\label{sec:conclusion}

We define the notion of a \plll reduced basis and give two algorithms to compute such bases.  Both
algorithms are polynomial time improvements of LLL and are based on the concept of deep insertions
as in Schnorr and Euchner's \dlll.  While the provable bounds of the achieved Hermite factor are not
better than for classical LLL -- in fact, for reduction parameter~$\delta = 1$, the existence of
critical bases shows that better lattice-independent bounds do not exist -- the
practical behavior
is much better than for classical LLL and they outperform BKZ-5.

It is striking to see that although our two algorithms to compute a \plll reduced basis only differ
in the strategy of choosing the insertion depth, their practical behavior is different. We
therefore believe that is might be worth to consider yet other strategies of choosing the
insertions.  Further an insertion can be seen as a special kind of permutation of the basis
vectors. Ensuring that an insertion only happens when it results in a proper decrease of the
potential of the basis ensures the polynomial running time of the algorithms.  This concept could be
generalized to other classes of permutations. The crucial point is the easy computation of the
change of the potential under the different permutations.

It is likely that the improvements of the $L^2$ algorithm~\cite{ng06} and the $\tilde{L^1}$
algorithm~\cite{NovocinSV11} can be used to improve the runtime of our \plll algorithm, in order to
achieve faster runtime. We leave this for future work.

\paragraph*{Acknowledgements}
This work was supported by CASED (\url{http://www.cased.de}). Michael Schneider is
supported by project BU~630/23-1 of the German Research Foundation (DFG). Urs Wagner and Felix
Fontein are supported by SNF grant no.~132256.

\end{document}